\documentclass[runningheads]{llncs}
\usepackage[utf8]{inputenc}
\usepackage{array}
\usepackage{textcomp}
\usepackage{graphicx}
\usepackage{wrapfig}
\usepackage{amsfonts}
\usepackage{amsmath}
\usepackage{amssymb}
\usepackage{stmaryrd}
\usepackage{hyperref}
\usepackage{setspace}
\usepackage{listings}
\usepackage{makeidx}  
\usepackage{skull}
\usepackage{comment}
\usepackage{ifthen}

\usepackage{breakurl}             

\newcommand{\img}{\operatorname{img}}

\newcommand{\len}{\operatorname{len}}

\newcommand{\concat}{\ensuremath{\!+\!\!+}}
\newcommand{\cons}{\ensuremath{::}}

\newcommand{\map}{\operatorname{map}}
\newcommand{\lex}{\ensuremath{\sqsubseteq}}
\newcommand{\nil}{\ensuremath{[\,]}}
\newcommand{\prefix}{\ensuremath{\preccurlyeq}}

\newboolean{arxiv}
\setboolean{arxiv}{true}

\newcommand{\arxiv}[1]{\ifthenelse{\boolean{arxiv}}{#1}{}}
\newcommand{\noarxiv}[1]{\ifthenelse{\boolean{arxiv}}{}{#1}}
\ifthenelse{\boolean{arxiv}}{\includecomment{arxivenv}}{\excludecomment{arxivenv}}

\renewcommand{\bf}{\normalfont \bfseries}

\newcommand{\xcite}[1]{\nolinebreak{ }\cite{#1}}

\title{An implementation of Deflate in Coq}
\author{Christoph-Simon Senjak and Martin Hofmann \email{\{christoph.senjak,hofmann\}@ifi.lmu.de} \institute{Ludwig-Maximilians-Universität\\ Munich, Germany}}
\titlerunning{Deflate in Coq}
\authorrunning{C. Senjak \& M. Hofmann}
\begin{document}
\maketitle

\begin{abstract}
  The widely-used compression format ``Deflate'' is defined in RFC 1951 and is based on prefix-free codings
  and backreferences. There are unclear points about the way these codings are specified, and several sources
  for confusion in the standard. We tried to fix this problem by giving a rigorous mathematical specification,
  which we formalized in Coq. We produced a verified implementation in Coq which achieves competitive
  performance on inputs of several megabytes. In this paper we present the several parts of our
  implementation: a fully verified implementation of canonical prefix-free codings, which can be used in other
  compression formats as well, and an elegant formalism for specifying sophisticated formats, which we used to
  implement both a compression and decompression algorithm in Coq which we formally prove inverse to each
  other -- the first time this has been achieved to our knowledge. The compatibility to other Deflate
  implementations can be shown empirically. We furthermore discuss some of the difficulties, specifically
  regarding memory and runtime requirements, and our approaches to overcome them.
\end{abstract}

\section{Introduction}
It is more and more recognized that traditional methods for maintenance of software security reach their
limits, and different approaches become inevitable\xcite{darpa}. At the same time, formal program verification
has reached a state where it becomes realistic to prove correctness of low-level system components and combine
them to prove the correctness of larger systems. A common pattern is to have a kernel that isolates parts of
software by putting them in sandboxes. This way, one gets strong security guarantees, while being able to use
unverified parts which might fail, but cannot access memory or resources outside their permissions. Examples
are the L4 verified kernel\xcite{Klein_AEMSKH_14} and the Quark browser\xcite{jang2012establishing}.

This is an important step towards fully verified software, but it is also desirable to verify the low-level
middleware. While for these components the adherence of access restrictions would be assured by an underlying
sandbox, functional correctness becomes the main concern. The CompCert compiler is such a project, and as
\cite{Leroy-Compcert-CACM} points out, a compiler bug can invalidate all guarantees obtained by formal
methods. The MiTLS\xcite{mitls} project implements TLS, and verifies cryptographic security properties. We
propose to add to this list a collection of compression formats; in this paper we look specifically at
Deflate\xcite{rfc1951}, which is a widely used standard for lossless general purpose compression. HTTP can
make use of it\xcite{rfc2616}, so does ZIP\xcite{zipspec633} and with it its several derived formats like Java
Archives (JAR) and Android Application Packages (APK). Source-code-tarballs are usually compressed with GZip,
which is a container around a Deflate stream\xcite{rfc1952}. Finally, TLS supports Deflate
compression\xcite{rfc3749}, though it is now discouraged due to the BREACH family of
exploits\xcite{breach}. Deflate compression can utilize Huffman codings and Backreferences as used in
Lempel-Ziv-Compression (both defined later), but none of them are mandatory: The way a given file can be
compressed is by no means unique, making it possible to use different compression algorithms. For example, the
\texttt{gzip(1)} utility has flags \texttt{-1} through \texttt{-9}, where \texttt{-9} yields the strongest but
slowest compression, while \texttt{-1} yields the weakest but fastest compression. Furthermore, there are
alternative implementations like Zopfli\xcite{zopfli}, which gains even better compression at the cost of
speed.

It is desirable to have some guarantees on data integrity, in the sense that the implementation itself will
not produce corrupted output. A common complaint at this point is that you can get this guarantee by just
re-defining your unverified implementations of compression, say $c$, and decompression, say $d$, by
\begin{eqnarray*}
  c'x &=& \left\{\begin{array}{rl} (\top, cx) & \mbox{ for } d(cx) = x \\ (\bot, x) & \mbox{ otherwise} \end{array}\right. \\
  d'x &=& \left\{\begin{array}{rl} dy & \mbox{ for } x = (\top, y) \\ y & \mbox{ for } x = (\bot, y)\end{array}\right.
\end{eqnarray*}
This works well as long as one only has to work with one computer architecture. However, for secure
long-term-archiving of important data, this is not sufficient: It is not clear that there will be working
processors being able to run our $d$ implementation in, say, 50 years; but a formal, mathematical,
human-readable specification of the actual data format being used can mitigate against such digital
obsolescence: The language of mathematics is universal. However, of course, this is a benchmark one should
keep in mind. We are currently still far away from this performance level, but we are sure our work can lead
to such a fast implementation, but not without lots of micro-optimization; for now the performance is
acceptable but not fast enough yet, we are working on making it better.

Of course, one needs {\it some} specification. Besides having to rely on some hardware specification, as
pointed out in\xcite{Klein_AEMSKH_14}, finding the right formal specification for software is not trivial. In
MiTLS\xcite{mitls}, an example about ``alert fragmentation'' is given, which might give an attacker the
possibility to change error codes by injection of a byte. This is standard compliant, but obviously not
intended. A rigorous formal specification of an informally stated standard must be carefully crafted, and we
consider our mathematical specification of Deflate as a contribution in this direction.

\subsection{Related Work}
To our best knowledge, this is the first verified pair of compression and decompression algorithms, and it is
practically usable, not just for toy examples. However, there have been several projects that are related. A
formalization of Huffman's algorithm can be found in\xcite{Blanchette:2009:PPM:1541694.1541701}
and\xcite{thery2004formalising}. As we will point out in Section \ref{section:DeflateCodings}, the codings
Deflate requires do not need to be Huffman codings, but they need to satisfy a canonicity condition. From the
general topic of data compression, there is a formalization of Shannon's theorems in Coq\xcite{Affeldt2014}.

There are two techniques in Coq that are commonly regarded as ``program extraction'': On the one hand, one can
explicitly write functions with Coq, and prove properties about them, and the extract them to OCaml and
Haskell. This is the method that is usually used. The complexity of the extracted algorithms can be estimated
easily, but reasoning about the algorithms is disconnected from the algorithms themselves. On the other hand,
it is possible to write constructive existence proofs and extract algorithms from these proofs directly. The
advantage of this approach is that only a proof has to be given, which is usually about as long as a proof
about an explicit algorithm, so the work only has to be done once. However, the disadvantage is that the
complexity of the generated algorithm is not always obvious, especially in the presence of tactics. We think
that this technique fits well especially for problems in which either the algorithm itself is complicated,
because it usually has lots of invariants and proofs of such an algorithm require extensive use of the inner
structure of the terms, or when the algorithm is trivial but the proofs are long. The case study
\cite{Nogin_writingconstructive}, albeit on a different topic (Myhill-Nerode), is an interesting source of
inspiration in that it distills general principles for improving efficiency of extracted programs which we
have integrated where applicable. In particular, these were
\begin{itemize}
\item to use expensive statements non-computationally, which we have done in large parts of the code.
\item to use existential quantifiers as memory, which we did, for example, in our proofs regarding strong
  decidability (see Section \ref{section:StrongUniquenessAndDecidability}).
\item to calculate values in advance, which we did, for example, for the value \texttt{fixed\_lit\_code}.
\item to turn loop invariants into induction statements, which is not directly applicable because Coq does not
  have imperative loops, but corresponds to Coq's induction measures, which give a clue about the
  computational complexity.
\end{itemize}
We use both extraction techniques in our code. Besides the use of recursion operators instead of pattern
matching, the extracted code is quite readable.

Our theory of parsers from Section \ref{section:StrongUniquenessAndDecidability} follows an idea similar to
\cite{berger2015program}, trying to produce parsers directly from proofs, as opposed to other approaches, for
example \cite{danielsson2010total}, which defines a formal semantic on parser combinators. Most of the
algorithms involved in parsing are short, and therefore, as we already said, using the second kind of program
extraction we mentioned was our method of choice for the largest part.

\subsection{Overview}
In summary, this paper provides a rigorous formal specification of Deflate and a reference implementation of
both a compression and decompression algorithm which have been formally verified against this specification
and tested against the ZLib.

This paper is organized as follows: In Section \ref{section:TheEncodingRelation}, we give a very brief
overview over several aspects of the Deflate standard. In Section \ref{section:DeflateCodings} we introduce
concepts needed to understand the encoding mechanism of Deflate that is mostly used, namely {\it Deflate
  codings}, a special kind of prefix-free codings, and prove several theorems about them. In Section
\ref{section:Backreferences}, we will introduce the concept of {\it backreferences} which is the second
compression mechanism besides prefix-free codings that can be used with Deflate. Section
\ref{section:StrongUniquenessAndDecidability} is about our mechanism of specifying and combining encoding
relations, and how one can gain programs from these. Section \ref{section:Compression} will introduce our
current approach for a verified compression algorithm. Finally, Section \ref{fwork} explains how our software
can be obtained, compiled and tested. \arxiv{

  The final publication is available at \href{http://link.springer.com}{link.springer.com}.} \noarxiv{We
  published a version of this paper with an appendix containing an elaborate example of the Deflate
  compression standard, some benchmarks, and some explanatory tables, on Arxiv.}

\section{The encoding relation}\label{section:TheEncodingRelation}
The main problem when verifying an implementation of a standard is that a specification must be given, and
this specification itself is axiomatic and cannot be formally verified. We address this problem in two
ways. First, we try to put the complexity of the implementation into the proofs, and make the specification as
simple as possible. The correctness of a specification should be ``clear'' by reading, or at least require
only a minimal amount of thinking. This was not always possible, because the Deflate standard is intricate; in
the cases when it was not possible, we tried to at least put the complexity into few definitions and reuse
them as often as possible. In fact, most of our definitions in {\tt EncodingRelation.v} should be easily
understandable when knowing the standard. In addition to that, we give some plausibility checks in the form of
little lemmas and examples which we formally prove. Secondly, we prove a decidability property for our
encoding relation which yields---by program extraction---a reference implementation that we can apply to
examples.  This way, the implementation becomes empirically correct. However, even if there was a pathological
example in which our specification is not compliant with other implementations, it would still describe a
formally proved lossless compression format, and every file that was compressed with one of our verified
compression algorithms could still be decompressed with every verified decompression algorithm.

On the toplevel, a stream compressed with Deflate is a sequence of blocks, each of which is a chunk of data
compressed with a specific method. There are three possible types of blocks: uncompressed blocks, which save a
small header and the clear text of the original, statically compressed blocks, which are compressed with
codings defined in the standard, and dynamically compressed blocks, which have codings in their header. Their
respective type is indicated by a two-bit header. Furthermore, a third bit indicates whether the block is the
last block. The bit-level details of the format are not important for this paper, most of the relational
definition can be found in the file {\tt EncodingRelation.v}. For clarity, we give an informal illustration of
the toplevel format:\\

\texttt{
  \begin{tabular}{lcl}
    Deflate & ::= & ('0' Block)* '1' Block ( '0' | '1' )* \\
    Block   & ::= & '00' UncompressedBlock | \\
            &     & '01' DynamicallyCompBl | \\
            &     & '10' StaticallyCompBl \\
    UncompressedBlock & ::= & length \~{}length bytes \\
    StaticallyCompBl & ::= & CompressedBlock(standard coding) \\
    DynamicallyCompBl & ::= & header coding CompressedBlock(coding)\\
    CompressedBlock(c) & ::= & [\^{}256]* 256 (encoded by c)
\end{tabular}}\\

\noindent Compressed blocks can contain backreferences -- instructions to copy already decompressed bytes to
the front -- which are allowed to point across the borders of blocks, see Section
\ref{section:Backreferences}. A decompression algorithm for such blocks must, besides being able to resolve
backreferences, be able to decompress the data according to two codings, where some of the codes have
additional suffixes of a number of bits defined in a table\arxiv{ (see Appendix
  \ref{appendix:tables})}. Additionally, for dynamically compressed blocks, the codings themselves, which are
saved as sequences of numbers (see Section \ref{section:DeflateCodings}), are compressed by a third
coding. This makes decompression of such blocks a lot harder than one would expect, and gives a broad vector
for common bugs like off-by-one-errors or misinterpretations of the standard. For example, notice that while
the first table from \arxiv{Appendix \ref{appendix:tables}}\noarxiv{the first table in Section 3.2.5 from the
  standard \cite{rfc1951}} looks quite ``continuous'', the codepoint 284 can only encode 30 code lengths,
which means that the suffixes 01111 and 11111 are illegal (this was actually a bug in an early version of our
specification). Due to the space restrictions, we will not get deeply into the standard in this paper, and
spare the readers the complicated parts as far as possible. For a deeper understanding, we give an elaborate
example in \arxiv{Appendix \ref{appendix:overview}}\noarxiv{the Arxiv-version of this paper} and
otherwise refer to\xcite{rfc1951}.

\section{Deflate Codings}\label{section:DeflateCodings}
Deflate codings are the heart of Deflate. Everything that is compressed in any way will be encoded by two
Deflate codings, even if the coding itself is not used to save memory (this will usually be the case for
statically compressed blocks which only utilize backreferences). In other literature, Deflate codings are also
called {\bf canonical prefix-free codings} -- ``canonical'' because of the result shown in Theorem
\ref{uniqueness}, ``prefix-free'' will be defined in Definition \ref{deflatecoding}. Sometimes people talk
about ``codes'' instead of ``codings''. However, in our terminology, a ``code'' is a sequence of bits from a
coding, and a ``coding'' is a map from an alphabet into codes. Though we call them Deflate codings, they are
also used in many other compression formats, like BZip2, and this part of our implementation can be reused.

It is well-known\xcite{huffman} that for every string $A\in{\cal A}^*$ over a finite alphabet ${\cal A}$,
there is a {\bf Huffman coding} $h : {\cal A} \rightarrow \{0,1\}^*$, which is a prefix-free coding such that
the concatenation of the characters' images $\operatorname{foldl} (\concat) \nil (\map h A)$ has minimal length. In
fact, this has already been formally proved\xcite{Blanchette:2009:PPM:1541694.1541701}. The standard
\cite{rfc1951} abuses terminology slightly by calling any not necessarily optimal prefix-free coding ``Huffman
coding''. This makes sense because, especially for statically compressed blocks, fixed, not necessarily
optimal encodings are used. On the other hand, the standard specifies canonical prefix-free codings which can
be uniquely reconstructed from the respective code lengths for each encoded character. These canonical codings
are referred to as Deflate codings. Therefore, instead of expensively saving a tree structure, it is
sufficient to save the sequence of code lengths for the encoded characters. Optimal Deflate codings are also
known as {\bf canonical Huffman codings}.

In any practical case, there will be a canonical ordering on ${\cal A}$, so from now on, let us silently
assume the alphabet ${\cal A}=\{0,\ldots,n-1\}$ for some $n\in\mathbb{N}$. We say a code $a$ is a {\bf prefix}
of $b$ and write $a\prefix b$, if there is a list $c\in\{0,1\}^*$ such that $a\concat\,c=b$. Notice that
$\prefix$ is reflexive, transitive and decidable. We denote the standard {\bf lexicographical ordering} on
$\{0,1\}^*$ by $\lex$. We have $\nil\lex a$ and $0\cons a \lex 1\cons b$ for all $a, b$ and $j\cons a \lex
j\cons b$ whenever $a \lex b$. It is easy to show that this is a decidable total ordering relation. We can now
make prefix-free codings unique. The code $\nil$ is used to denote that the corresponding element of $\cal A$
does not occur. This is consistent with the standard that uses the code length $0$ to denote this.

\begin{definition}\label{deflatecoding}
  A {\bf Deflate coding} is a coding $\lceil\cdot\rceil:{\cal A}\rightarrow\{0,1\}^*$ which satisfies the
  following conditions:
  \begin{itemize}
  \item[1.] $\lceil\cdot\rceil$ is prefix-free, except that there may be codes of length
    zero: \[\forall_{a,b}.(a\neq b\wedge \lceil a\rceil\neq\nil)\rightarrow\lceil a\rceil\not\prefix \lceil
    b\rceil\]
  \item[2.] Shorter codes lexicographically precede longer codes: \[\forall_{a,b}.\len \lceil a\rceil < \len \lceil
    b\rceil\rightarrow\lceil a\rceil\lex \lceil b\rceil\]
  \item[3.] Codes of the same length are ordered lexicographically according to the order of the characters
    they encode: \[\forall_{a,b}.(\len\lceil a\rceil=\len\lceil b\rceil\wedge a\le b)\rightarrow \lceil
    a\rceil\lex \lceil b\rceil\]
  \item[4.] For every code, all lexicographically smaller bit sequences of the same length are prefixed by
    some code: \[\forall_{a\in{\cal A},l\in\{0,1\}^+}.(l\lex\lceil a\rceil\wedge\len l=\len\lceil
    a\rceil)\rightarrow \exists_b.\lceil b\rceil\neq\nil\wedge\lceil b\rceil\prefix l\]
  \end{itemize}
\end{definition}
These axioms are our proposed formalization of the informal specification
in\xcite{rfc1951}, which states: ``The Huffman codes used for each alphabet in the `deflate' format have two
additional rules:
\begin{itemize}
  \item All codes of a given bit length have lexicographically consecutive values, in the same order as the
    symbols they represent;
  \item Shorter codes lexicographically precede longer codes.''
\end{itemize}
Notice that prefix-codes as given by their code lengths do not necessarily correspond to optimal,
i.e. Huffman, codes. For example, the Deflate coding
\[  0 \rightarrow [0], 1 \rightarrow [1, 0, 0], 2 \rightarrow [1, 0, 1], 3 \rightarrow [1, 1, 0] \]
is clearly not a Huffman coding, since for every case it would apply to, we could also use
\[  0 \rightarrow [0], 1 \rightarrow [1, 0], 2 \rightarrow [1, 1, 1], 3 \rightarrow [1, 1, 0] \]
which will always be better. Unique recoverability, however, holds true for all Deflate codings irrespective
of optimality.

Axiom 3 is weaker than the first axiom from\xcite{rfc1951}, as it does not postulate the consecutivity of the
values, which is ensured by axiom 4: Assuming you have characters $a<b$ such that $\len\lceil
a\rceil=\len\lceil b\rceil$, and there is a $l\in\{0,1\}^{\len\lceil a\rceil}$ such that $\lceil a\rceil\lex
l\lex\lceil b\rceil$, then by axiom 4 there is a $d$ such that $\lceil d\rceil\prefix l$. Trivially, $\lceil
a\rceil\lex\lceil d\rceil$, therefore by axiom 2, it follows that $\lceil d\rceil = l$. That is, if there is a
code of length $\len\lceil a\rceil$ between $\lceil a\rceil$ and $\lceil b\rceil$, then it is the image of a
character. Therefore, the values of codes of the same length are lexicographically consecutive.

Furthermore, consider our non-optimal coding from above: It has the lengths $0\rightarrow 1,1\rightarrow
3,2\rightarrow 3,3\rightarrow 3$, and satisfies our axioms 1-3, and additionally, the codes of the same length
have lexicographically consecutive values. But the same holds for the coding
\[  0 \rightarrow [0], 1 \rightarrow [1, 0, 1], 2 \rightarrow [1, 1, 0], 3 \rightarrow [1, 1, 1] \]
However, in this coding, there is a ``gap'' between the codes of different lengths, namely between $[0]$ and
$[1,0,1]$, and that is why it violates our axiom 4: The list $[1,0,0]$ is lexicographically smaller than
$[1,0,1]$, but it has no prefix.

We can show that Deflate codings are uniquely determined by their code lengths:

\begin{theorem}[uniqueness]\label{uniqueness} Let $\lceil\cdot\rceil,\lfloor\cdot\rfloor:{\cal A}\rightarrow\{0,1\}^*$ be two Deflate codings, such that $\forall_{x\in{\cal A}}.\len\lceil x\rceil=\len\lfloor x\rfloor$. Then $\forall_{x\in{\cal A}}.\lceil x\rceil=\lfloor x\rfloor$.
\end{theorem}
\begin{proof}
Equality of codings is obviously decidable, therefore we can do a proof by contradiction, without using the
law of excluded middle as an axiom. So assume there were two distinct deflate codings $\lceil\cdot\rceil$ and
$\lfloor\cdot\rfloor$ with $\len\lceil \cdot\rceil=\len\lfloor \cdot\rfloor$. Then there must exist $n, m$
such that $\lceil n\rceil = \min_\lex\{\lceil x\rceil\mid \lceil x\rceil \neq \lfloor x\rfloor\}$ and $\lfloor
m\rfloor = \min_\lex\{\lfloor x\rfloor\mid \lceil x\rceil \neq \lfloor x\rfloor\}$. If $\len{\lceil
  n\rceil}>\len{\lfloor m\rfloor}$, then also $\len{\lceil n\rceil}>\len{\lceil m\rceil}$, and by our axiom 2,
${\lceil m\rceil}\lex{\lceil n\rceil}$. But $m$ was chosen minimally. Symmetric for $\len{\lceil
  n\rceil}>\len{\lfloor m\rfloor}$. Therefore, $\len{\lceil n\rceil}=\len{\lfloor m\rfloor}$. Also, $\lfloor
m\rfloor\neq\nil$, because otherwise $0=\len{\lfloor m\rfloor}=\len{\lceil n\rceil}$, so $\lceil
n\rceil=\nil$, and so $\lceil m\rceil=\lfloor m\rfloor$, which contradicts our assumption on the choice of
$m$. Analogous, $\lceil n\rceil\neq\nil$. By totality of $\lex$, we know that $\lceil n\rceil\lex\lfloor
m\rfloor\vee\lfloor m\rfloor\lex\lceil n\rceil$. Both cases are symmetric, so without loss of generality
assume $\lceil n\rceil\lex\lfloor m\rfloor$. Now, by axiom 4, we know that some $b$ exists, such that $\lfloor
b\rfloor\prefix\lceil n\rceil$, therefore by axiom 2, $\lfloor b\rfloor\lex\lfloor m\rfloor$, and thus, by the
minimality of $m$, either $b = m$ or $\lfloor b\rfloor = \lceil b\rceil$. $b=m$ would imply $\lceil m\rceil =
\lfloor m\rfloor$, which contradicts our choice of $m$. But $\lfloor b\rfloor = \lceil b\rceil$ would imply
$\lceil b\rceil\prefix\lceil n\rceil$, which contradicts our axiom 1.
\end{proof}
This theorem is proved as Lemma \texttt{uniqueness} in \texttt{DeflateCoding.v}. While uniqueness is a
desirable property, it does not give us the guarantee that, for every configuration of lengths, there actually
is a Deflate coding. And in fact, there isn't: Trivially, there is no Deflate coding that has three codes of
length $1$. It is desirable to have a simple criterion on the list of code lengths, that can be efficiently
checked, before creating the actual coding.

Indeed, the well-known Kraft inequality\xcite{1056818} furnishes such a criterion. It asserts that a
prefix-free coding with code lengths $k_0,\ldots,k_{N-1}$ exists iff $$\sum_{i=0}^{N-1} 2^{-k_i}\le 1$$
Deflate codings may, however, have $k_i = 0$ if the corresponding character does not occur. Moreover, we want
to extract an algorithm from this proof, so we have to prove it constructively.
\begin{theorem}[extended\_kraft\_ineq]\label{kraftd}
  Let $\lceil\cdot\rceil:{\cal A}\rightarrow\{0,1\}^*$ be a Deflate coding. Then
  $$\mathop{\sum\limits_{i\in{\cal A}}}_{\lceil i\rceil\neq\nil}2^{-\len\lceil i\rceil}\le 1$$ Equality holds
  if and only if there is some $k\in{\cal A}$ such that $\lceil k\rceil\in\{1\}^+$.
\end{theorem}
This is formally proven as {\tt extended\_kraft\_ineq} in {\tt DeflateCoding.v}. The most important theorem
regarding Deflate codings is:
\begin{theorem}[existence]\label{existence}
  Let $l:{\cal A}\rightarrow\mathbb{N}$ be such that $$\mathop{\sum\limits_{i\in{\cal A}}}_{l(i)\neq
    0}2^{-l(i)}\le 1$$ Then there is a Deflate coding $\lceil\cdot\rceil:{\cal A}\rightarrow\{0,1\}^*$ such
  that $l x = \len\lceil x\rceil$.
\end{theorem}
For the proof, we introduce the notation $[n]^k := [\underbrace{n,\ldots,n}_{k\times}]$.

\begin{proof}
  Let $\lesssim$ be the right-to-left lexicographical ordering relation on $\mathbb{N}$, defined
  by \[\forall_{mqo}.q<o\rightarrow(q,m)\lesssim(o,m)\]
  \[\forall_{m_1,m_2,n_1,n_2}.m_1<m_2\rightarrow(n_1,m_1)\lesssim(n_2,m_2)\]
  Now let $R=L:=\operatorname{sortBy} (\lesssim) (\map (\lambda_k(k,lk)) [0,\ldots,n-1])$, $S=\nil$,
  $cx=\nil$. We will do a recursion on tuples $(S,c,R)$, maintaining the following invariants:
  \begin{enumerate}
  \item \label{inv1} If a pair is not in the list of already processed pairs $S$, then it is in the list of
    remaining pairs $R$, and the corresponding code is empty \[\forall_q.(q,\len(c(q)))\not\in S\rightarrow
    (c(q)=\nil\wedge(q,l(q))\in R)\]
  \item \label{inv2} $L$ contains the elements of $S$ and $R$ \[(\operatorname{rev} S)\concat R = L\]
  \item \label{inv3} Either $S$ is empty, or the code corresponding to its first element is lexicographically
    larger than every code in the current coding \[S=\nil\vee\forall_q.c(q)\lex c(\pi_1(\operatorname{first}
    S))\]
  \end{enumerate}
  Furthermore, $c$ will be a Deflate coding at every step. The decreasing element will be $R$, which will
  become shorter at every step. We first handle the simple cases:
  \begin{itemize}
    \item For the initial values $(\nil,\lambda_x\nil,L)$, the invariants are easy to prove.
    \item For $R=\nil$, we have $\operatorname{rev} S=L$ by \ref{inv2} and therefore, either $c=\lambda_x\nil$ if
      $L=\nil$, or $\forall_q.(q,\len(c(q)))\in L$ by \ref{inv1}, and therefore, $c$ is the desired coding.
    \item For $R=(q,0)::R'$, $S$ can only contain elements of the form $(\_,0)$. We proceed with
      $((q,0)::S,\lambda_x\nil,R')$. All invariants are trivially preserved.
    \item For $R=(q,1+l)::R'$ and $S=\nil$ or $S=(r,0)::S'$, we set $c'(x)=[0]^{1+l}$ for $x=q$, and
      $c'(x)=\nil$ otherwise. We proceed with $((q,1+l)::S,c',R')$. The invariants are easy to show. It is
      easy to show that $c'$ is a Deflate coding.
  \end{itemize}
  The most general case is $R=(q,1+l)::R'$ and $S=(r,1+m)::S'$; let the intermediate Deflate coding $c$ be
  given. We have \[\mathop{\sum\limits_{i\in{\cal A}}}_{c(i)\neq\nil}2^{-\len(c(i))}<2^{-l-1}+
  \mathop{\sum\limits_{i\in{\cal A}}}_{c(i)\neq\nil}2^{-\len(c(i))}\le
  \mathop{\sum\limits_{i\in{\cal A}}}_{li\neq 0}2^{-l(i)}\le 1\] By Theorem \ref{kraftd}, $[1]^{1+m}\not\in\img
  c$, and therefore, we can find a fresh code $d'$ of length $1+m$. Let $d=d'\concat[0]^{l-m}$ and
  set \[c'(x):=\left\{\begin{array}{cl}d&\mbox{for }x=q\\c(x)&\mbox{otherwise}\end{array}\right.\] We have to
  show that $c'$ is a Deflate coding. The axioms 2 and 3 are easy. For axiom 4, assume $x\neq\nil$ and $x\lex
  c'(q)$. If $x\lex c'(r)$, the claim follows by axiom 4 for $r$. Otherwise, by totality $c'(r)\lex x$. If
  $x\lex d'$, by the minimality of $d'$ follows $x=c'(r)$. If $d'\lex x$, trivially, $d'\prefix c'(q)$. Axiom
  4 holds. For axiom 2, it is sufficient to show that no other non-$\nil$ code prefixes $d$. Consider a code
  $e\prefix d$. As all codes are shorter or of equal length than $d'$, $e\prefix d'$. But then, either
  $e\prefix c(r)$, or $c(r)\lex e$. Contradiction. Therefore, we can proceed with $((q,1+l)::S,c',r')$.
\end{proof}
This is proved as Lemma {\tt existence} in {\tt DeflateCoding.v}. From this, we can extract an algorithm that
calculates a coding from a sequence of lengths. For a better understanding of the algorithm proposed here, we
consider the following length function as an example:
$$ l : 0 \rightarrow 2; 1 \rightarrow 1; 2 \rightarrow 3; 3 \rightarrow 3; 4 \rightarrow 0 $$
We first have to sort the graph of this function according to the $\lesssim$ ordering.
$$ [(4, 0), (1,1), (0, 2), (2, 3), (3, 3)] $$
Then, the following six steps are necessary to generate the coding.
{\footnotesize
\begin{center}
\begin{tabular}{|>\centering m{0.70cm}||*2{>\centering m{2cm}|}|*5{m{1cm}<{\centering}|}}
\hline
 Step & R & S & c(0) & c(1) & c(2) & c(3) & c(4) \\ \hline \hline 
 0 & [(4,~0), (1,~1), (0,~2), (2,~3), (3,~3)] & \nil & \nil & \nil & \nil & \nil & \nil \\ \hline
 1 & [(1,~1), (0,~2), (2,~3), (3,~3)] & [(4,~0)] & \nil & \nil & \nil & \nil & \nil \\ \hline
 2 & [(0,~2), (2,~3), (3,~3)] & [(1,~1), (4,~0)] & \nil & [0] & \nil & \nil & \nil \\ \hline
 3 & [(2,~3), (3,~3)] & [(0,~2), (1,~1), (4,~0)] & [1,0] & [0] & \nil & \nil & \nil \\ \hline
 4 & [(3,~3)] & [(2,~3), (0,~2), (1,~1), (4,~0)] & [1,0] & [0] & [1,1,0] & \nil & \nil \\ \hline
 5 & \nil & [(3,~3), (2,~3), (0,~2), (1,~1), (4,~0)] & [1,0] & [0] & [1,1,0] & [1,1,1] & \nil \\ \hline
\end{tabular}
\end{center}}

The final values of $c$ are, in fact, a Deflate coding. The main difference to the algorithm in the standard
\cite{rfc1951} is that we sort the character/length pairs and then incrementally generate the coding, while
their algorithm counts the occurrences of every non-zero code length first, determines their lexicographically
smallest code, and then increases these codes by one for each occurring character. In our case, that means
that it would first generate the function $a:1 \rightarrow 1; 2 \rightarrow 1; 3 \rightarrow 2 $ and $0$
otherwise, which counts the lengths, and then define $$b(m) = \sum_{j=0}^{m-1} 2^ja(j)$$ which gets the
numerical value for the lexicographically smallest code of every length when viewed as binary number with the
most significant bit being the leftmost bit. In our case, this is $1 \rightarrow 0; 2\rightarrow 2;
3\rightarrow 6$. Then $$ c(n) = b(l(n)) + |\{r < n\mid l(r)=l(n)\}| $$ meaning $c(0) = b(2) = 10_{(2)}$, $c(1)
= b(1) = 0_{(2)}$, $c(2) = b(3) = 110_{(2)}$, $c(3) = b(3) + 1 = 111_{(2)}$ which is consistent with the
algorithm presented here. The algorithm described in the standard\xcite{rfc1951} is more desirable for
practical purposes, as it can make use of low-level machine instructions like bit shifting. On the other hand,
notice that our algorithm is purely functional.

\section{Backreferences}\label{section:Backreferences}
Files usually contain lots of repetitions. A canonical example are C files which contain lots of {\tt
  \#include} statements, or Java files which contain lots of {\tt import} statements in the beginning. Deflate
can remove these repetitions, as long as they are not more than 32 KiB\footnote{Kibibyte: $2^{10}$ Byte} apart
from each other. The mechanism uses backreferences, as found in Lempel-Ziv-compression. An extension of the
backreference mechanism also allows for run length encoding (see below). A backreference is an instruction to
copy parts of already decompressed data to the front, so duplicate strings have to be saved only once. They
are represented as a pair $\left<l,d\right>$ of a length $l$ and a distance $d$. The length is the number of
bytes to be copied, the distance is the number of bytes in the backbuffer that has to be skipped before
copying. Similar mechanisms are used in other compression formats, so our implementation can probably be used
for them, too.

The resolution (decompression) of such backreferences in an imperative context is trivial, but uses lots of
invariants that make it hard to prove correct. In a purely functional context, it is non-trivial to find data
structures that are fast enough. We decided to stick with purely functional algorithms, as they can be
verified directly using Coq, and optimization of purely functional programs is interesting for its own
sake. In our current verified implementation, this is the slowest part. \arxiv{The benchmarks in Appendix
  \ref{benchmarks} support this claim.} We already have figured out an algorithm with better performance, but
we are not yet done proving it formally correct; we will not get deeper into this algorithm in this
paper.

Assuming we wanted to compress the string

\begin{equation}
  \texttt{ananas\_banana\_batata}
\end{equation}
we could shorten it with backreferences to
\begin{equation}\label{ananas1}
  \texttt{ananas\_b}\left<5,8\right>\left<3,7\right>\texttt{tata}
\end{equation}
An intuitive algorithm to resolve such a backreference uses a loop that decreases the length and copies one
byte from the backbuffer to the front each time (the example is written in Java; notice that this algorithm,
while intuitive, is not suitable for actual use in a decompression program, because you usually do not know
the length of the output in advance, and hence cannot allocate an array of the proper length):
\begin{lstlisting}[language=Java]
int resolve (int l, int d, int index, byte[] output) {
    while (l > 0) {
      output[index] = output[index-d];
      index = index + 1; l = l - 1; }
    return index; }
\end{lstlisting}
This intuitive algorithm works when $l>d$, and results in a repetition of already written bytes -- which is
what run length encoding would do. Therefore, Deflate explicitly allows $l>d$, allowing us to shorten
(\ref{ananas1}) even further:
\begin{equation}\label{ananas2}
  \texttt{an} \left<3,2\right> \texttt{s\_b} \left<5,8\right> \left<3,7\right> \texttt{t} \left<3,2\right>
\end{equation}
More directly, the string $\texttt{aaaaaaaargh!}$ can be compressed as
$\texttt{a}\left<7,1\right>\texttt{rgh!}$, which essentially is run length encoding.

As already mentioned, the efficient resolution of backreferences in a purely functional manner was a lot
harder than we expected. An imperative implementation can utilize the fact that the distances are limited by
32 KiB, and use a 32 KiB ringbuffer in form of an array that is periodically iterated and updated in parallel
to the file-I/O. This uses stateful operations on an array, and has complicated invariants.

\subsection{A verified backreference-resolver}\label{section:rbr}
The obvious approach to do this in a purely functional way is using a map-like structure instead of an array
as a ring buffer. The best possible approach we found uses an exponential list
\begin{lstlisting}
Inductive ExpList (A : Set) : Set :=
| Enil : ExpList A
| Econs1 : A -> ExpList (A * A) -> ExpList A
| Econs2 : A -> A -> ExpList (A * A) -> ExpList A.
\end{lstlisting}
This takes into account that -- in our experience -- most backreferences tend to be ``near'', that is, have
small distances, and such elements can be accessed faster. We could just save our whole history in one {\tt
  ExpList} that we always pass around, without performance penalty. However, this will take a lot of memory
which we do not need, as backreferences are limited to 32 KiB. We use another technique which we call {\bf
  Queue of Doom}: We save two {\tt ExpList}s and memorize how many elements are in them. The front {\tt
  ExpList} is filled until it contains 32 KiB. If a backreference is resolved, and its distance is larger than
the amount of bytes saved in the front {\tt ExpList}, it is looked up in the back {\tt ExpList}. Now, if the
front {\tt ExpList} is 32 KiB large, the front {\tt ExpList} becomes the new back {\tt ExpList}, a new empty
front {\tt ExpList} is allocated, and the former back {\tt ExpList} will be doomed to be eaten by the garbage
collector. The following is an illustration of filling such a queue of doom, the {\tt ExpList}s are denoted as
lists, and their size is -- for illustration -- only 3:
\begin{eqnarray*}
  \mbox{start} & \nil & \nil \\
  \mbox{push 1} & [1] & \nil \\
  \mbox{push 2} & [2; 1] & \nil \\
  \mbox{push 3} & [3; 2; 1] & \nil \\
  \mbox{push 4} & [4] & [3; 2; 1] \qquad \nil \rightarrow \skull\\
  \mbox{push 5} & [5; 4] & [3; 2; 1]  \\
  \mbox{push 6} & [6; 5; 4] & [3; 2; 1] \\
  \mbox{push 7} & [7] & [6; 5; 4] \qquad [3; 2; 1] \rightarrow \skull \\
\end{eqnarray*}

The advantage of this algorithm is that we have a fully verified implementation in {\tt
  EfficientDecompress.v}. The disadvantage is that while it does not perform badly, it still does not have
satisfactory performance, taking several minutes\arxiv{ -- as you can see in Appendix \ref{benchmarks}}. We
are currently working on better algorithms. One such algorithm which is purely functional can be found in the
file {\tt NoRBR/BackRefs.hs} in the software directory, see Section \ref{fwork}. Another such algorithm which
utilizes diffarrays \cite{diffarray} aka Parray \cite{vafeiadis2013adjustable}, and resembles an imperative
resolution procedure, can be found in the file {\tt NoRBR/Back\-Ref\-With\-Diff\-Array.hs}. Both perform well,
and we are currently working on verifying them.

\section{Strong Decidability and Strong Uniqueness}\label{section:StrongUniquenessAndDecidability}
So far we showed how we implemented several aspects of the standard. However, this was a very high-level view:
We still need to combine the parts we implemented in the way specified in\xcite{rfc1951}. This is a lot less
trivial than it might sound: A compressed block is associated with two codings, a ``literal/length'' coding,
and a ``distance'' coding. The ``literal/length'' coding contains codes for raw bytes, a code for the end of
the block, and ``length'' codes, which initialize a backreference, and can have suffixes of several
bits\arxiv{ as specified in the table in Appendix \ref{appendix:tables}}. Such length codes and their suffix
must be followed by a ``distance'' code which can also have a suffix. Dynamically compressed blocks have an
additional header with the code-length sequences for these two codings (which are sufficient for
reconstruction of the codings, as proved in Section \ref{section:DeflateCodings}). However, these sequences
are themselves compressed by yet another mechanism that -- besides Huffman-coding -- allows for
run-length-encoding. Therefore, a third coding must be specified in the header, the ``code-length
coding''. Uncompressed blocks, on the other hand, must start at a byte-boundary, which means that when
specifying, we cannot even forget the underlying byte sequence and just work on a sequence of bits.

We could have written a decompression function as specification, but there are several possible algorithms to
do so, which we would have to prove equivalent. We decided that a relational specification is clearer and
easier to use, and probably also easier to port to other systems (Isabelle, Agda, Minlog) if desired. We
defined two properties that such relations must have to be suitable for parsing, which we will define in this
section.

While efficiency in runtime and memory are desirable properties, the most important property of a lossless
compression format is the guarantee that for any given data $d$,
$\operatorname{decompress}(\operatorname{compress} d)=d$, which is what our final implementation
guarantees. While most container formats have some checksum or error correction code, Deflate itself does not
have mechanisms to cope with data corruption due to hardware failures and transcription errors, therefore a
formal discussion of these is outside the scope of this paper; research in this direction can be found for
example in \cite{affeldt2015formalization}.

We will work with relations of the form {\tt OutputType -> InputType -> Prop}. The final relation is called
{\tt DeflateEncodes}.

Left-Uniqueness (``injectivity'') can be formalized as $\forall_{a,b,l}.R\,a\,l\rightarrow R\,b\,l\rightarrow
a=b$. However, when reading from a stream, it must always be clear when to ``stop'' reading, which essentially
means that given an input $l$ such that $R\,a\,l$, it cannot be extended:
$\forall_{a,b,l,l'}.R\,a\,l\rightarrow R\,b\,(l\concat l')\rightarrow l' = \nil$. We proved that these two
properties together are equivalent to the following property, which we call {\it strong uniqueness}:
\begin{eqnarray*}
&\operatorname{StrongUnique}(R):\Leftrightarrow \\ & \forall_{a,b,l_a,l_a',l_b,l_b'}. l_a\concat l_a'=l_b\concat l_b'\rightarrow R a l_a \rightarrow R b l_b \rightarrow a = b \wedge l_a = l_b
\end{eqnarray*}
This is formally proved as {\tt StrongUniqueLemma} in {\tt StrongDec.v}. While strong uniqueness gives us
uniqueness of a prefix, provided that it exists, we need an additional property that states that it is
actually decidable whether such a prefix exists, which we call {\it strong decidability}:
$$
\operatorname{StrongDec}(R):\Leftrightarrow \forall l. (\lambda_X.X\vee\neg X) (\exists_{a, l', l''}. l=l'\concat l''\wedge R a l')
$$
All existences are constructive: If a prefix exists, then we can calculate it. Therefore, proving strong
decidability yields a parser for the respective relation. Conversely, if you can write and verify a parser for
it, then existence follows.

Strong decidability and strong uniqueness reflect the obvious type of a verified decoder: If a relation
satisfies both properties, it is well-suited for parsing. Indeed, for $R$ being our formalization of the
Deflate standard, we give a formal proof of $\operatorname{StrongDec}(R)$ which is such that the extracted
decoding function constitutes a usable reference implementation in the sense that it can successfully
decompress files of several megabytes. We can combine such relations in a way similar to parser monads, a
bind-like combinator can be defined that first applies the first relation, and then the second relation:
\[ Q \gg\!\!=_c R := \mu_\xi (\forall_{b_q,b_r,a_q,a_r}. Q\,b_q\,a_q \rightarrow R\,b_q\,b_r\,a_r \rightarrow \xi\,(c\,b_q\,b_r)\,(a_q +\!+ a_r)) \]
This combinator preserves strong uniqueness and decidability. More complicated combinators can be built from
it. This makes it is easy to replace parts of strong decidability proofs and optimize them, and makes the
implementation modular. This way we could benchmark optimizations before verifying them (by using {\tt admit},
for example), which made programming much easier.

The definitions can be found in {\tt StrongDec.v}, most proofs for our encoding relation can be found in {\tt
  EncodingRelationProperties.v}

We think that our overall theory of such grammars and parsers is usable for many other data formats: It should
be usable whenever parsing does not need to be interactive in the sense that it must produce answers to
requests (like in many bidirectional network protocols). But despite this drawback, it should be applicable in
many practical situations, and is very simple.

\section{Compression}\label{section:Compression}
Compression is by no means unique, and depends on the desired trade off between speed and compression ratio. We
implemented an algorithm that does not yet utilize optimal Huffman codings, but only searches for possible
backreferences, and saves everything as statically compressed blocks. Especially for ASCII texts this is
usually a disadvantage, and we plan to include this into the algorithm in the future to gain better
compression results. The algorithm calculates a hashsum of every three read bytes and saves their position in
a hash table which has queues of doom as buckets. This follows a recommendation from \cite{rfc1952}, adapted
to the purely functional paradigm. The implementation can be found in {\tt Compress.v}.

\section{Conclusion}\label{fwork}
Our contribution is a complete mathematical formalization of Deflate. We formalized the proofs in Coq, such
that an implementation of a decompression algorithm in Haskell can be extracted. We tested this implementation
against some datasets, and observed that it is compatible with other implementations of Deflate. We
implemented a simple compression algorithm and a decompression algorithm, both fully verified against the
specification, with reasonable speed.

The project's source code can be found under
\\\url{http://www2.tcs.ifi.lmu.de/~senjak/fm2016/deflate.tar.gz}. For build instructions, see {\tt
  README.txt}. It works under Coq 8.4pl6, and GHC version 7.10.3, but most of the code should be portable
across versions. We also plan to maintain our GitHub-repository at
\url{https://github.com/dasuxullebt/deflate} in the future.

We gave a flexible, modular and simple way of specifying grammars and using these specifications to create
stream parsers. Our project shows that program extraction from proofs and performance are not a
contradiction. We already developed two not-yet verified faster algorithms to resolve backreferences, one of
which is purely functional, which we will formally verify in the future. While we believe that there is still
potential for optimization of our Coq code, we hope to use our specification to create a verified
implementation in C, using the Verified Software Toolchain \cite{programlogics}.

\nocite{Schwichtenberg2013-SCHMFC}

\begin{arxivenv}
\begin{appendix} 
\section{Benchmarks}\label{benchmarks}
Notice that in our tests, we added GZip headers so we could easily decompress it with {\tt gunzip(1)}; this
part is not formally verified, as it is not part of Deflate -- but as it is just the adding of a small header
and checksum, a formal verification would not add much value, especially as we only use this part to interact
with unverified software. The results of the benchmark can be found in the table below. We can furthermore
extract a decompression algorithm from the parsability proofs. This is useful for testing, but also
interesting for its own sake. In\xcite{Nogin_writingconstructive}, some principles for writing proofs with
efficient extracted algorithms are given, which we mostly followed. We made the observation that relying on
lazy evaluation, as done by Haskell, gives these principles for free to a certain extent. Though we initially
hoped that this was sufficient to get a usable implementation, it turned out to be only usable for very small
datasets. We then chose a combined ``top-down'' approach, in which we tried to identify the slowest parts of
our extracted program, and optimize these. The modular relational design of our implementation proved as a
useful property, as well as Coq's possibility to override definitions using {\tt Extract
  Constant}. Furthermore, this enables us to replace parts of the program with unverified code first, and
verify that code afterwards, if it performs well. The performance-critical part of the decompression is the
resolution of backreferences. We therefore show the performance of the extracted decompression algorithm
without resolution of backreferences in our benchmarks in the table below (the program will then yield a
sequence of bytes and backreferences). The compression rate of {\tt kennedy.xls} is especially good, while the
times we need for compression and decompression are especially bad. This suggests that the compressed version
of {\tt kennedy.xls} contains lots of backreferences. (The unverified algorithm we mentioned in Section
\ref{section:rbr} only takes seconds to decompress {\tt kennedy.xls} and is purely functional, and we already
have an informal correctness proof, but we are not done proving it formally correct.) The benchmarks are for
the Canterbury Corpus \cite{canterbury} on an {\tt Intel(R) Core(TM) i7-4770 CPU @ 3.40GHz}:

\noindent
\begin{tabular}{|c||c|c|c|c|c|}
\hline
File & Original & Comp- & Compress & Decompress & Decompress\\
~ & Bytes & ressed & Time & (No Back- & \\
~ & ~ & Bytes & ~ & references) & ~ \\
\hline \hline
{\tt alice29.txt} & 152089 & 118126 & 2m51.480s & 0m5.047s & 5m13.108s \\ \hline
{\tt asyoulik.txt} & 125179 & 97187 & 1m56.700s & 0m4.207s & 4m23.135s \\ \hline
{\tt cp.html} & 24603 & 15139 & 0m6.907s & 0m0.934s & 0m39.143s \\ \hline
{\tt fields.c} & 11150 & 6325 & 0m1.735s & 0m0.576s & 0m4.101s \\ \hline
{\tt grammar.lsp} & 3721 & 2013 & 0m0.573s & 0m0.417s & 0m0.641s \\ \hline
{\tt kennedy.xls} & 1029744 & 439956 & 60m59.977s & 0m18.212s & 44m12.826s \\ \hline
{\tt lcet10.txt} & 426754 & 327304 & 27m18.671s & 0m13.391s & 19m12.085s \\ \hline
{\tt plrabn12.txt} & 481861 & 389913 & 40m13.934s & 0m16.114s & 21m7.902s \\ \hline
{\tt ptt5} & 513216 & 269382 & 49m16.218s & 0m13.327s & 23m22.609s \\ \hline
{\tt sum} & 38240 & 21703 & 0m22.869s & 0m1.270s & 1m58.054s \\ \hline
{\tt xargs.1} & 4227 & 2656 & 0m0.847s & 0m0.425s & 0m1.481s \\ \hline
\end{tabular}

\section{An Overview of Deflate}\label{appendix:overview}
We give a short informal overview of correct Deflate streams, to show you the complexity of the format, and in
the hope that it will make it easier to follow our definitions and relations. {\bf Notice that we are
  describing an existing and widespread standard here. Especially, this standard was not made by us.} We are
giving this overview so you do not have to read the actual standard. There are many parts which seem
overcomplicated, but that is probably due to the fact that this is a grown standard. To clarify our
terminology, we say a {\em character} is an element from an alphabet, a {\em codepoint} is a number that is
encoded in some dataset and may stand for either a character or some instructional control structure, a {\em
  coding} is a function that assigns bit sequences to codepoints, and a {\em code} is a bit sequence which is
associated with some codepoint through a coding.

Deflate streams can make use of three techniques of compression: prefix-free coding (as in Huffman codes), run
length encoding and backreferences as found in Lempel-Ziv-compression. The latter two use the same mechanism,
as described in Section \ref{section:Backreferences}. Furthermore, Deflate streams are byte streams, which are
streams of values from $0$ to $255$. With such byte streams, one associates bit streams by concatenating the
bytes LSB (least significant bit first), regardless of how they are actually sent. This is necessary, because
most Deflate modes operate conceptually on the bit level.

On top of this bit stream, the data is sliced into blocks which may be compressed. A block starts with a bit
that indicates whether it is the last block, and two further bits indicating whether the block is
``statically'' compressed, that is, with fixed codings defined in the standard, or ``dynamically'' compressed,
where the codings must be saved, or uncompressed.

For an uncompressed block, the bits up to the next byte boundary are ignored, then a 16 bit integer followed
by its bitwise complement are saved byte aligned. It denotes the number of bytes the block
contains. Uncompressed blocks cannot contain backreferences. The advantage of the byte aligned layout of
uncompressed blocks is that it allows for the use of byte-wise access, e. g. \texttt{sendfile(2)}. On the
formal level this brings the extra difficulty that Deflate streams cannot be described as a formal grammar on
a bit sequence without knowing the byte boundaries.

Compressed blocks start immediately after the three header bits. Statically compressed blocks have predefined
codings, and therefore, the compressed data immediately follows the header bits. Even when the actual
compression does not utilize Huffman codings to save memory directly (which will usually be the case for
statically compressed blocks), two prefix-free codings are needed to encode backreferences: A coding does not
only encode the 256 byte-values, but up to 286 (288 with 2 forbidden) characters, of which one, 256, is used
as end mark, and the values above 256 are used to denote backreferences. If the decoder encounters a code for
such a character, a certain number of additional bits is read, from which the length of this backreference is
calculated. Then, using another coding, a value from 0 to 29 is read, and additional bits, which determine the
distance of that backreference. The numbers of actual bits for characters can be looked up in a table
specified in the standard \cite{rfc1951}.

Dynamically compressed blocks get another header defining the two Deflate codings. The codings are saved as
sequences of numbers, as formalized in Section \ref{section:DeflateCodings}. This way is similar in other
compression standards that utilize prefix-free codings, like bzip2. These sequences are themselves compressed,
and another header is needed to save their coding.

For clarity, let us consider a small example. As we have to deal with three layers of compression, it is not
always clear what a code, a coding and a character is. For this example, we add indices to the words to denote
which layer they are from. A code$_n$ is a sequence of bits for a codepoint$_n$. A codepoint$_n$ is a number
assigned to either a character$_n$ or some special instruction on that level. A coding$_n$ is a deflate coding
for codepoints$_n$. Raw bytes are characters$_0$. We want to compress the string

$$ \texttt{ananas\_banana\_batata} $$

Firstly, as we want to compress, we need an end sign (which gets the code\-point$_0$ 256), which we will
denote as $\emptyset$.  Since this string has a lot of repetitions, we can use backreferences. A backreference
is a pair $\left<l,d\right>$ of length and distance, which can be seen as an instruction to copy $l$ bytes
from a backbuffer of decompressed data, beginning at the $d$-th last position, to the front, in ascending
order, such that the copied bytes may be bytes that are written during the resolution of this backreference,
hence allowing for both deduplication and run length encoding. In our case, we can add two backreferences.
$$ \texttt{an} \left<3,2\right> \texttt{s\_b} \left<5,8\right> \left<3,7\right> \texttt{t} \left<3,2\right> \emptyset $$
The codepoint$_0$ for length $3$ is $257$, and for $5$ it is $259$. They do not have suffixes. The codepoint$_0$ for the
distance $2$ is $1$ with no suffix, for $7$ and $8$ it is $5$, and it has a single bit as suffix, which indicates
whether it stands for $7$ or $8$. We write $a^l$ to denote that $a$ is a literal/length codepoint$_0$, with an index
denoting the corresponding character$_0$ if any, and $a^d$ to denote that it is a distance codepoint$_0$. We furthermore
put suffix bit sequences in brackets. Then we get
$$ 97^l_a\ 110^l_n\ 257^l 1^d\ 115^l_s\ 95^l_{\_}\ 98^l_b\ 259^l 5^d (1)\ 257^l\ 5^d (0)\ 116^l_t\ 257^l 1^d\ 256^l_\emptyset $$
The frequencies of literal/length codepoints$_0$ are
\[ 95\times 1; 97\times 1; 98\times 1; 110\times 1; 115\times 1; 116\times 1; 256\times 1; 257\times 3; 259\times 1\]
The frequencies of distance codepoints$_0$ are
\[1\times 2; 5\times 2\]
The optimal deflate codings$_0$ (as defined in Section \ref{section:DeflateCodings}) are
\[ 95\rightarrow 1100; 97\rightarrow 010; 98\rightarrow 011; 110\rightarrow 100; 115\rightarrow 1101 \]
\[ 116\rightarrow 101; 256\rightarrow 1110; 257\rightarrow 00; 259\rightarrow 1111 \]
and
\[ 1\rightarrow 0; 5\rightarrow 1 \]

To clarify the terminology, note that e. g. character$_0$ \texttt{a} has codepoint$_0$ $010$ under the given
coding$_0$. The reason for introducing the concept of ``codepoints$_0$'' is that the alphabets for lengths and
characters$_0$ are merged: Every character$_0$ has an assigned codepoint$_0$, but not every codepoint$_0$ has
a character$_0$, e. g. the codepoint$_0$ $257$ indicates a backreference, but still has the code$_0$ $00$. Our
message can therefore be encoded by the following sequence of bits (spaces are included for clarity):

\begin{lstlisting}
010 100 00 0 1101 1100 011 1111 1 1 00 1 0 101 00 0 1110
\end{lstlisting}

As proved in Section \ref{section:DeflateCodings}, it is sufficient to save the lengths, which is done as a
run length encoded list, where length $0$ means that the corresponding codepoint$_0$ does not occur. We use a
semicolon to separate the literal/length coding$_0$ from the distance coding$_0$. Both lists are not separated
in the actual file, and it is even allowed that run-length-encoding-instructions spread across their
border. What part of the unfolded list belongs to which coding is specified in another header defined later.
\[ \underbrace{0,\ldots,0}_{95\times}, 4, 0, 3, 3, \underbrace{0,\ldots,0}_{11\times}, 3, \underbrace{0,\ldots,0}_{138\times}, 0, 4, 2, 0, 4; 0, 1, 0, 1 \]
This list will itself be compressed, thus, the lengths of codes$_0$ become characters$_1$. Notice that due to
a header described later, we can cut off all characters$_1$ after the last nonzero character$_1$ of both
sequences. The maximum length that is allowed for a code$_0$ in deflate is $15$. Deflate uses the
codepoints$_1$ $16,17,18$ for its run length encoding. Specifically, $17$ and $18$ are for repeated
zeroes. $17$ gets a $3$ bit suffix ranging from $3$ to $10$, and $18$ gets a $7$ bit suffix, ranging from $11$
to $138$. These suffixes are least-significant-bit first. The former sequence therefore becomes
\[ 18(0010101), 4, 0, 3, 3, 18(0000000), 3, 17(100), \]
\[ 4, 3, 18(1111111), 0, 4, 2, 0, 4; 0, 1, 0, 1 \]
Now, the frequencies of codepoints$_1$ are
\[ 1\times 2; 2\times 1; 3\times 4; 4\times 4; 17\times 1; 18\times 2 \]
Therefore, the optimal coding$_1$ is
\[ 1 \rightarrow 1110; 2 \rightarrow 1111; 3 \rightarrow 00; 4 \rightarrow 01; 17 \rightarrow 101; 18 \rightarrow 110 \]
The sequence of code$_0$ lengths can therefore be saved as
\begin{lstlisting}
110 0010101 01 100 00 00 110 0000000 00 101 100 01 00
110 1111111 100 01 1111 100 01 100 1110 101 000 1110
\end{lstlisting}

We now have to save the coding$_1$, and again, it is sufficient to save the code$_1$ lengths. These code$_1$
lengths for the $19$ codepoints$_1$ are saved as $3$ bit least-significant-bit first numbers, but in the
following order: $16,\allowbreak 17,\allowbreak 18,\allowbreak 0,\allowbreak 8,\allowbreak 7,\allowbreak
9,\allowbreak 6,\allowbreak 10,\allowbreak 5,\allowbreak 11,\allowbreak 4,\allowbreak 12,\allowbreak
3,\allowbreak 13,\allowbreak 2,\allowbreak 14,\allowbreak 1,\allowbreak 15$. Again, the codepoint$_2$ $0$
denotes that the corresponding codepoint$_1$ does not occur. We can furthermore cut off the codepoint$_2$ for
the last code$_1$ length (in the given order), $15$, which is $0$ in our example, due to a header described
later. The sequence of codepoints$_2$ therefore becomes

\begin{lstlisting}
000 110 110 110 000 000 000 000 000
000 000 010 000 010 000 001 000 001
\end{lstlisting}

We now come to the aforementioned header that in particular allows us to economize trailing zeroes. We need
the number of literal/length codepoints$_0$ and distance codepoints$_0$ saved in the former sequence, and the
number of saved codepoints$_2$. These are 260, 6 and 18, respectively. The first one is saved as a 5 bit
number ranging from $257$ to $286$ (the values $287$ and $288$ are forbidden), the second one is saved as a 5
bit number ranging from $1$ to $32$, the third one is saved as a 4 bit number ranging from $4$ to
$19$. Therefore, this header becomes

\begin{lstlisting}
11000 10100 0111
\end{lstlisting}

With three additional header bits, denoting that what follows is the last block, and that it is a dynamically
compressed block, (and with 7 additional bits to fill up the byte in line \ref{lst:lxfill}) we get
\begin{lstlisting}[label={lst:shortExample},escapeinside={@}{@}]
@\label{lst:lxbh}@1 0 1
@\label{lst:lxhhh}@11000 10100 0111
 
@\label{lst:lxclc1}@000 110 110 110 000 000 000 000 000
@\label{lst:lxclc2}@000 000 010 000 010 000 001 000 001

@\label{lst:lxlld1}@110 0010101 01 100 00 00 110 0000000 00 101 100 01 00
@\label{lst:lxlld2}@110 1111111 100 01 1111 100 01 100 1110 101 000 1110

@\label{lst:lxdata}@010 100 00 0 1101 1100 011 1111 1 1 00 1 0 101 00 0 1110

@\label{lst:lxfill}@0000000
\end{lstlisting}
Of course, this example is constructed for instructional purposes, and the compressed message is longer than
the original text. However, Deflate also supports statically compressed blocks, which are good for repetitive
files. Those use a fixed coding$_0$ which is completely described in the standard\xcite{rfc1951}. Its relevant
part for our example is the following:
\begin{eqnarray*}
& 95^l\rightarrow 10001111; 97^l\rightarrow 10010001; 98^l\rightarrow 10010010; 110^l\rightarrow 10011110; &\\
& 115^l\rightarrow 10100011; 116^l\rightarrow 10100100; 256^l\rightarrow 0000000; 257^l\rightarrow 0000001; &\\
& 259^l\rightarrow 0000011; 1^d \rightarrow 00001; 5^d \rightarrow 00101 &
\end{eqnarray*}
With the three header bits, and 4 additional padding bits to fill the byte, the compressed file is
\begin{lstlisting}
1 1 0

10010001 10011110 0000001 00001 10100011 10001111 10010010
0000011 00101 1 0000001 00101 0 10100100 0000001 00001
0000000

0000
\end{lstlisting}
which is, in fact, slightly shorter than the original string. Since we did this manually, we wanted to check
whether it is actually correct, so we wrote a little program that uses the zlib to do that. We share this
program with you: In our software directory, see Section \ref{fwork}, it is called {\tt bits.cpp}. You can
just pipe the bit lists from our listings into this program (notice that under Linux, when you copy-paste them
manually into stdin, you need to press Ctrl-D afterwards to enforce {\tt EOF}), and notice that they, in fact,
decompress to our original string.

\section{Tables from the RFC}\label{appendix:tables}

The following table from \cite{rfc1951} assigns the number of extra (suffix) bits, and the range of lengths
that can be encoded, to length codes (literal/length codes $>256$):

{\footnotesize
  \begin{tabular}{|c|c|c||c|c|c||c|c|c|}
    \hline
     & Extra &          &      & Extra &        &      & Extra & \\
Code & Bits & Length(s) & Code & Bits & Lengths & Code & Bits & Length(s) \\
\hline
 257 & 0 &   3   &    267 &  1 &  15,16 &    277  & 4 &  67-82\\
 258 & 0 &   4   &    268 &  1 &  17,18 &    278  & 4 &  83-98\\
 259 & 0 &   5   &    269 &  2 &  19-22 &    279  & 4 &  99-114\\
 260 & 0 &   6   &    270 &  2 &  23-26 &    280  & 4 & 115-130\\
 261 & 0 &   7   &    271 &  2 &  27-30 &    281  & 5 & 131-162\\
 262 & 0 &   8   &    272 &  2 &  31-34 &    282  & 5 & 163-194\\
 263 & 0 &   9   &    273 &  3 &  35-42 &    283  & 5 & 195-226\\
 264 & 0 &  10   &    274 &  3 &  43-50 &    284  & 5 & 227-257\\
 265 & 1 &11,12  &    275 &  3 &  51-58 &    285  & 0 &   258\\
 266 & 1 &13,14  &    276 &  3 &  59-66 &         &   & \\ \hline
\end{tabular}}

The following table from \cite{rfc1951} assigns the number of extra (suffix) bits, and the range of distances
that can be encoded, to distance codes:

{\footnotesize
  \begin{tabular}{|c|c|c||c|c|c||c|c|c|}
    \hline
     & Extra &          &      & Extra &        &      & Extra & \\
Code & Bits & Distance(s) & Code & Bits & Distance(s) & Code & Bits & Distance(s) \\
\hline
               0 & 0 &  1  &  10 & 4 &   33-48 &  20  & 9 & 1025-1536\\
               1 & 0 &  2  &  11 & 4 &   49-64 &  21  & 9 & 1537-2048\\
               2 & 0 &  3  &  12 & 5 &   65-96 &  22 & 10 & 2049-3072\\
               3 & 0 &  4  &  13 & 5 &   97-128 &  23 & 10 & 3073-4096\\
               4 & 1 & 5,6 &  14 & 6 &  129-192 &  24 & 11 & 4097-6144\\
               5 & 1 & 7,8 &  15 & 6 &  193-256 &  25 & 11 & 6145-8192\\
               6 & 2 & 9-12 &  16 & 7 &  257-384 &  26 & 12 & 8193-12288\\
               7 & 2 & 13-16 &  17 & 7 &  385-512 &  27 & 12 & 12289-16384\\
               8 & 3 & 17-24 &  18 & 8 &  513-768 &  28 & 13 & 16385-24576\\
               9 & 3 & 25-32 &  19 & 8 & 769-1024 &  29 & 13 & 24577-32768\\ \hline
\end{tabular}}

\end{appendix}
\end{arxivenv}

\bibliographystyle{splncs03}
\bibliography{rfc,codings,compcert,RickettsRJTL14}
\end{document}